\documentclass[superscriptaddress,aps,pra,reprint,longbibliography]{revtex4-1}
\usepackage{etex}
\usepackage{graphicx}
\usepackage{caption}
\usepackage{subcaption}
\usepackage{mathrsfs}
\usepackage{amsfonts}
\usepackage{dsfont}
\usepackage{times}
\usepackage{amsmath}
\usepackage{amsthm}
\usepackage{amssymb}
\usepackage{leftidx}
\usepackage{tikz}
\usepackage{tikz-network}
\usepackage{color}
\usepackage[all]{xy}
\usepackage[ bookmarks=true, colorlinks, linkcolor=blue, urlcolor=blue, citecolor=blue, plainpages=false, pdfpagelabels, final, breaklinks=true ]{hyperref}

\usepackage{mathtools}
\usepackage{array}
\usepackage{multirow}
\usepackage{bbold}
\usepackage{ragged2e}%

\newcommand*{\cl}[1]{{\mathcal{#1}}}
\newcommand*{\bb}[1]{{\mathbb{#1}}}

\newcommand*{\tn}[1]{{\textnormal{#1}}}
\newcommand{\T}{\mbox{$\textnormal{Tr}$}}

\newcommand*{\eps}{\varepsilon}
\newcommand{\pr}{\mathrm{Pr}}
\newcommand{\vol}{\mathrm{vol}}
\newcommand{\area}{\mathrm{area}}

\newcommand{\xto}{\xrightarrow}

\newcolumntype{L}[1]{>{\raggedright\let\newline\\\arraybackslash\hspace{0pt}}m{#1}}
\newcolumntype{C}[1]{>{\centering\let\newline\\\arraybackslash\hspace{0pt}}m{#1}}
\newcolumntype{R}[1]{>{\raggedleft\let\newline\\\arraybackslash\hspace{0pt}}m{#1}}

\theoremstyle{plain}
\newtheorem{theorem}{Theorem}[section]
\newtheorem{prop}[theorem]{Proposition}
\newtheorem{lemma}[theorem]{Lemma}
\newtheorem{cor}[theorem]{Corollary}

\newtheorem{theoremintro}{Theorem}
\newtheorem{corollaryintro}{Corollary}
\newtheorem*{conj}{Conjecture}

\theoremstyle{definition}

\theoremstyle{remark}
\newtheorem{remark}[theorem]{Remark}

\DeclareRobustCommand{\orcidicon}{%
	\begin{tikzpicture}
	\draw[lime, fill=lime] (0,0) 
	circle [radius=0.16] 
	node[white] {{\fontfamily{qag}\selectfont \tiny ID}};
	\draw[white, fill=white] (-0.0625,0.095) 
	circle [radius=0.007];
	\end{tikzpicture}
	\hspace{-2mm}
}
\foreach \x in {A, ..., Z}{%
	\expandafter\xdef\csname orcid\x\endcsname{\noexpand\href{https://orcid.org/\csname orcidauthor\x\endcsname}{\noexpand\orcidicon}}
}

\begin{document}
\title{Concentration of Measure Phenomena for Quantum States on a Higher Dimensional Equator}
\author{Kabgyun Jeong\orcidA{}}
\email{kgjeong6@snu.ac.kr}
\affiliation{Institute of Computer Technology, College of Engineering, Seoul National University, Seoul 08826, Korea}
\affiliation{Team QST, Seoul National University, Seoul 08826, Korea}

\date{\today}

\begin{abstract}
We revisit L\'{e}vy's lemma, a widely used analytical tool in quantum information theory. Concentration inequalities quantify the phenomenon in which Lipschitz observables concentrate around a median or mean, and serve as fundamental analytical tools across information theory, statistical physics, and learning theory. In particular, L\'{e}vy's lemma provides a crucial framework for describing functionals on pure quantum states, with applications in quantum entanglement and quantum statistical query learning. In this work, we isolate the hyper-equatorial part of the standard spherical concentration argument. The resulting estimate is a L\'{e}vy-type bound for Lipschitz functions on a fixed hyperequator, with the natural dimension parameter $d-1$. We also formulate the accompanying geometric localization in terms of neighborhoods of the boundary, hyperequator, and a codimension-two antipodal great subsphere. This viewpoint clarifies the structure of the usual proof and points to the measure-theoretic formulation needed for sharper constant-level statements.
\end{abstract}
\maketitle

\section{Introduction}
The concentration of measure phenomenon is one of the characteristic geometric features of high-dimensional spaces. In its spherical form, the principle says that if $x$ is sampled uniformly at random from a high-dimensional sphere and a function $f$ has a Lipschitz constant, then $f(x)$ is almost constant except on a set whose measure decays exponentially in the dimension.  Equivalently, the spheres form a normal L\'{e}vy family.  This idea, going back to L\'{e}vy and developed in asymptotic convex geometry by Milman, Gromov-Milman, Talagrand, Ledoux, and others, is now a standard language for comparing geometry, probability, and high-dimensional analysis~\cite{L, M1, GM1, MS, T1, T2, Le, BLM, V, FGM, M2, M3, Ma, P}.

The same phenomenon has become a central tool in quantum information theory.  A pure state in a finite-dimensional Hilbert space is a point on a complex unit sphere, and Haar-random states or Haar-random unitaries are governed by the same type of high-dimensional geometry. L\'{e}vy-type estimates are used to prove that generic bipartite pure states are nearly maximally entangled, that typical reduced states are thermal, that random unitaries approximately randomize states, and that unitary designs can replace Haar randomness in several large-deviation arguments~\cite{HLSW, HLW, PSW, GLTZ, Low, HL, BHH} including an approximate QRAM construction for the noisy linear problem~\cite{J23}. More recently, concentration arguments have also appeared in quantum learning and variational quantum algorithms, for example in barren plateau results and in quantum statistical query models for states, processes, and Haar-random unitaries~\cite{GFE, McC, AGY, WD, Ang}. These applications typically need the exponential dependence on dimension, while the precise geometric location of the concentration and the constants in the exponent are often treated only up to universal factors.

This study revisits the spherical proof of L\'{e}vy's lemma from this geometric point of view.  We keep the notation used throughout the paper: $\cl{S}^d$ denotes the closed $d$-dimensional unit ball in $\bb{R}^d$, $\partial \cl{S}^d$ denotes its boundary sphere, $\mu$ denotes the corresponding uniform probability measure on the space under discussion, and
\begin{equation}
E^{d-2}:=\partial\cl{S}^d\cap H_{d-1}^0
\end{equation}
denotes a hyperequator, where $H_{d-1}^0$ is a hyperplane through the origin. The usual L\'{e}vy picture says that most of the boundary measure is contained in a thin neighborhood of a median equator.  We ask how this familiar equatorial localization can be separated from the proof and then followed by an analogous localization around a codimension-two great subsphere inside the hyperequator.

Our first main result is the following hyper-equatorial concentration estimate, proved as Theorem \ref{antipodal'}.

\begin{theoremintro} \label{antipodal}
Let $f:E^{d-2}\to \bb{R}$ be a $\kappa$-Lipschitz function on the hyperequator $E^{d-2}=\partial\cl{E}^{d-1}$.  If $x\in E^{d-2}$ is chosen uniformly at random, then, for every $\eps>0$,
\begin{equation} \label{eq.antipodal}
\pr\left[|f(x)-m_f|\ge\eps\right]
\le C_1\exp\left(\frac{-c(d-1)\eps^2}{\kappa^2}\right),
\end{equation}
where $m_f$ is a median of $f$ and $C_1,c>0$ are universal constants.
\end{theoremintro}

The proof follows the same two classical routes as L\'{e}vy's lemma: one route uses the Brunn--Minkowski inequality and Euclidean neighborhoods, and the other uses the spherical isoperimetric inequality.  The point is not to replace L\'{e}vy's lemma, but to separate the successive geometric localizations that are usually folded into one asymptotic statement: the volume of the ball is near the boundary, the surface measure of the boundary is near a hyperequator, and the induced measure on the hyperequator admits a further localization around a codimension-two antipodal great subsphere.

The second structural point is a caution about this final localization.  The antipodal great subsphere is terminal only relative to the coordinate localization scheme fixed above; a further localization would require additional structure, such as a new choice of coordinates or a new metric-measure space.

\begin{remark} \label{finalCM'}
Let $\cl{A}^{d-3}=E^{d-2}\cap H_{d-2}^0$ denote the codimension-two antipodal great subsphere determined by two orthogonal coordinate hyperplanes.  The localization statement below should be understood in terms of neighborhoods of $\cl{A}^{d-3}$, not as a literal equality between a probability measure and a subset.  In particular, we do not claim that no subset $\cl{X}$ or no further concentration function can exist; any such further statement must specify an additional metric-measure structure.
\end{remark}

The geometric evidence for this final localization is developed in Section \ref{sec.CM}.  In particular, the following three consequences are derived there independently as Corollaries \ref{cor.mu1}, \ref{cor.mu3}, and \ref{cor.mu4}.

\begin{corollaryintro} \label{corin.mu1}
Let $\cl{S}^d$ be a $d$-dimensional hypersphere. Suppose that $\mu$ is a uniform probability measure in the hypersphere $\cl{S}^d$. Then, for sufficiently large $d$,
\begin{equation}
\lim_{d\to\infty}\mu\left(\left\{x\in\cl{S}^d:\operatorname{dist}(x,\partial\cl{S}^d)\le \eta\right\}\right)=1
\end{equation}
for every fixed $\eta>0$.
\end{corollaryintro}

\begin{corollaryintro} \label{corin.mu3}
Let $\partial\cl{S}^d$ be a $d-1$-dimensional hypersurface. Suppose that $\mu$ is a uniform probability measure on the hypersurface $\partial\cl{S}^d$. Then, for sufficiently large $d$,
\begin{equation}
\lim_{d\to\infty}\mu\left(\left\{x\in\partial\cl{S}^d:\operatorname{dist}(x,E^{d-2})\le \eta\right\}\right)=1,
\end{equation}
for every fixed $\eta>0$, where $E^{d-2}:=\partial\cl{E}^{d-1}$ denotes an arbitrary hyperequator lying on the hypersurface $\partial\cl{S}^d$.
\end{corollaryintro}

\begin{corollaryintro} \label{corin.mu4}
Let $E^{d-2}$ be a $d-2$-dimensional hyperequator. Suppose that $\mu$ is a uniform probability measure on the hyperequator $E^{d-2}$. Then, for sufficiently large $d$,
\begin{equation}
\lim_{d\to\infty}\mu\left(\left\{x\in E^{d-2}:\operatorname{dist}(x,\cl{A}^{d-3})\le \eta\right\}\right)=1,
\end{equation}
for every fixed $\eta>0$, where $\cl{A}^{d-3}$ denotes a codimension-one antipodal great subsphere lying on the arbitrary hyperequator $E^{d-2}$.
\end{corollaryintro}

These statements lead to the following conjectural form of the final concentration locus.

\begin{conj} \label{conj}
Let $\mu$ be a uniform probability measure on a sphere $\cl{S}^{d}$. For sufficiently large $d$,
\begin{equation} \label{eq.conj}
\lim_{d\to\infty}\mu\left(\left(\cl{A}^{d-3}\right)_\eta\right)=1
\end{equation}
for every fixed $\eta>0$, where $\left(\cl{A}^{d-3}\right)_\eta$ denotes the $\eta$-neighborhood of an antipodal great subsphere lying on an arbitrary hyperequator $E^{d-2}$.
\end{conj}

For a uniform probability measure $\mu$ on the $d$-dimensional
hypersphere $\cl{S}^d$ and sufficiently large $d$, the proposed
localization scheme is summarized by
\begin{widetext}
$$\xymatrix{\cl{S}^d
\ar@{.>}[dr]_{\tiny{\mathrm{Corollary}~\ref{cor.mu2}}~~~~{}~~~~{}~~~~{}~~~~{}
~~~~{}~~~~{}~~~~{}~~~~{}~~~~{}~~~~{}~~~~{}~~~~{}}
&\xto{{}~~~~{}~~~~{}~~~~{}\tiny{\mathrm{Corollary}~\ref{corin.mu1}}
{}~~~~{}~~~~{}~~~~{}}
{}~~~~{}~~~~{}~~~~{}\partial\cl{S}^d{}~~~~{}~~~~{}~~~~{}
\xto{\tiny{{}~~~~{}~~~~{}~~~~{}\mathrm{Corollary}~\ref{corin.mu3}}
{}~~~~{}~~~~{}~~~~{}}& E^{d-2} \\
&\cl{E}^{d-1}
\ar@{.>}[ru]_{{}~~~~{}~~~~{}~~~~{}~~~~{}~~~~{}~~~~{}~~~~{}
~~~~{}~~~~{}~~~~{}~~~~{}
\tiny{\mathrm{Corollary}~\ref{cor.mu2}+
\mathrm{Lemma}~\ref{volannu}}}&}
\xto{\tiny{{}~~~~{}~~~~{}~~~~{}\mathrm{Theorem}~\ref{antipodal}
{}~~~~{}~~~~{}~~~~{}}}\cl{A}^{d-3}.
$$
\end{widetext}
The physical implication of the concentration of measure phenomena can be given by Fig.~\ref{Figure-1}. Also, we notice that a quantum state $\varrho$ obeys $\T\varrho=1$, i.e., it always (geometrically) forms a \emph{unit sphere} at any dimension $d$. 

\begin{figure*}
\centering
\includegraphics[width=15.5cm]{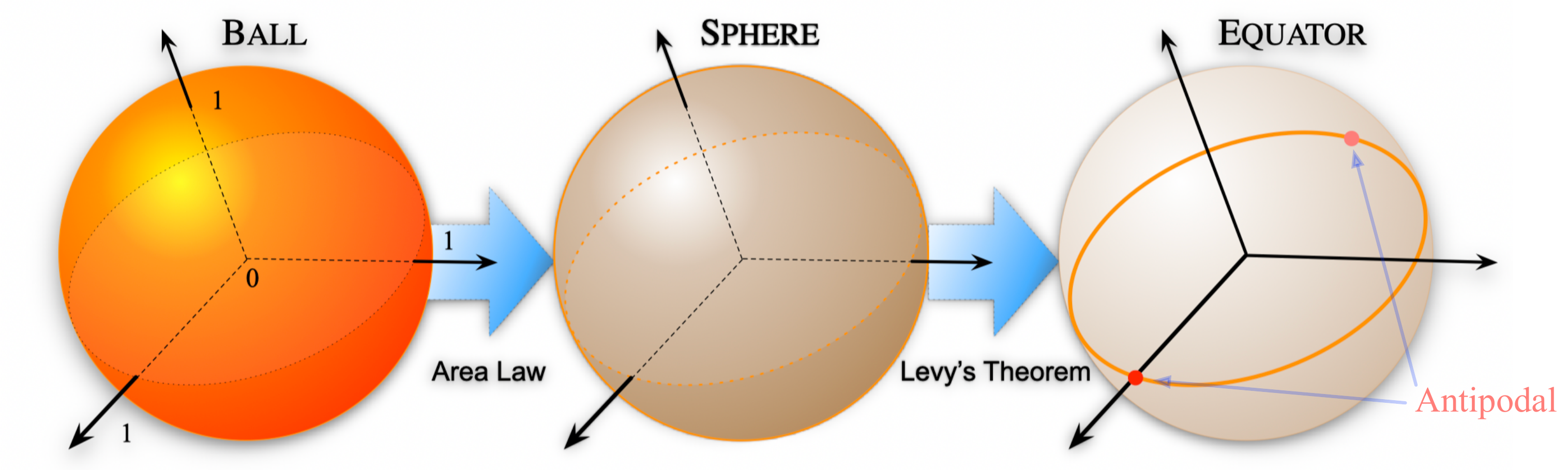}
\caption{Graphical representation on the phenomena of the concentration of measure for the higher dimensional quantum states: the area law corresponds to the holographic principle in high-energy physics.}
\label{Figure-1}
\end{figure*}

\subsection{Organization of the paper}
Section \ref{sec.levy} recalls the standard proof of L\'{e}vy's lemma, including the Brunn--Minkowski argument and the spherical
isoperimetric formulation. Section \ref{sec.main} proves the main hyper-equatorial estimate, first by the Lipschitz-function approach and then by the corresponding isoperimetric argument.  Section \ref{sec.CM} records the geometric concentration estimates for the volume, boundary, hyperequator, and antipodal locus of a high-dimensional hypersphere, and uses them to motivate Conjecture \ref{conj}. Section \ref{sec.conclusion} summarizes the main message and discusses possible directions for sharpening the result.

\subsection{Literature review}
\paragraph{Random states and random unitaries.}
The quantum information use of L\'{e}vy's lemma begins with the observation that Haar-random pure states are ordinary random points on a high-dimensional sphere.  Hayden, Leung, and Winter used this to show that random bipartite states and random subspaces display generic near-maximal entanglement~\cite{HLW}. Hayden, Leung, Shor, and Winter used related random-unitary concentration methods for approximate private quantum channels and state randomization~\cite{HLSW}. Canonical
typicality and the entanglement-based foundations of statistical mechanics use the same mechanism to show that small subsystems of typical constrained pure states look thermal~\cite{PSW, GLTZ}. The modern replacement of Haar randomness by efficient pseudorandom ensembles is treated through unitary designs and random circuits~\cite{Low, HL, BHH}. Additionally, an efficient quantum algorithm, for the noisy linear problem (NLP) or learning-with-error (LWE) problem under approximate quantum random access memory (QRAM), is proposed~\cite{J23}.

\paragraph{Concentration inequalities.}
The general concentration framework used here is the one developed in asymptotic geometric analysis and probability: L\'{e}vy's spherical isoperimetric inequality, the Gromov--Milman theory of L\'{e}vy families, Talagrand's inequalities, and the nonasymptotic concentration viewpoint~\cite{L, GM1, GM2, MS, T1, T2, Le, BLM, V, P}.  In recent quantum information theory, the same estimates underlie statements about typical entanglement, the rarity of useful highly entangled states for measurement-based computation, concentration on constrained state manifolds, barren plateaus in variational quantum circuits, and query lower bounds for learning Haar-random or design-like processes~\cite{GFE, MGE, McC, AGY, WD, Ang}.

\section{L\'{e}vy's theorem revisit} \label{sec.levy}

We begin by recalling well-known results on L\'{e}vy's theorem~\cite{L, Le, MS}.

\begin{theorem}[L\'{e}vy] \label{a.levy}
Let $F$ be a function $F:\partial\cl{S}^d\to \bb{R}$ defined on the $d$-dimensional sphere boundary $\partial\cl{S}^d$. Suppose that a point $x\in \partial\cl{S}^d$ is chosen uniformly at random. Then, for all $\delta>0$,
\begin{equation} \label{eq.a.levy}
\pr\left[|F(x)-\bb{E}(F)|\ge\delta\right]
\le C_1\exp\left(\frac{-C_2d\delta^2}{\kappa^2}\right),
\end{equation}
where $\kappa:=\sup|\nabla F|$ is the Lipschitz constant of $F$, and $C_1$, $C_2$ are positive constants.
\end{theorem}

Above theorem can be derived by followings: The Brun-Minkowski inequality~\cite{G, M, Ga} states that, for any nonempty compact set $A,B\subset\bb{R}^d$, $\vol(A)^{1/d}+\vol(B)^{1/d}\le\vol(A+B)^{1/d}$ where $A+B:=\{a+b|a\in A, b\in B\}$ denotes the Minkowski sum of $A$ and $B$.

\begin{lemma} \label{BMineq}
Let $A,B$ are any nonempty compact sets in $\bb{R}^d$. Then
\begin{equation} \label{eq.BMineq}
\vol\left(\frac{A+B}{2}\right)\ge\sqrt{\vol(A)\vol(B)},
\end{equation}
where $A+B:=\{a+b|a\in A, b\in B\}$ denotes the Minkowski sum of $A$ and $B$.
\end{lemma}

\begin{proof}
By Brun-Minkowski inequality and setting $d=1$,
\begin{eqnarray}
\vol\left(\frac{A+B}{2}\right)^{1/d}
&\ge&\vol\left(\frac{A}{2}\right)^{1/d}+\vol\left(\frac{B}{2}\right)^{1/d} \nonumber\\
&=&\frac{1}{2}\left(\vol(A)^{1/d}+\vol(B)^{1/d}\right) \\
&\ge&\left(\vol(A)\vol(B)\right)^{1/2d}. \nonumber
\end{eqnarray}
\end{proof}

Let $A$ be a nonempty compact subset of $\partial\cl{S}^d$. Define $A':=\{\gamma x|x\in A, \gamma\in[0,1]\}$ in $\cl{S}^d$. Then
\begin{equation}
\pr[A]=\mu(A'):=\frac{\vol(A')}{\vol(\cl{S}^d)}.
\end{equation}
Let $B$ be a set of $\partial\cl{S}^d\backslash(A\cup A_\eps)$, where $A_\eps$ denotes the $\eps$-neighborhood of $A$ such that the set of all $x\in\partial\cl{S}^d$ whose Euclidean norm $\|\cdot\|$ to $A$ is at most $\eps$. If $\eps\in[0,1]$ and $\pr[A]\ge1/2$, then, for all $a\in A$ and $b\in B$, $\|a-b\|\ge\eps$.

\begin{lemma} \label{euclid.dist}
Let $A'$ and $B'$ are subsets in $\cl{S}^d$. Then for any $x'\in A'$ and $y'\in B'$
\begin{equation} \label{eq.euclid.dist}
\left\|\frac{x'+y'}{2}\right\|\le1-\frac{\eps^2}{8}.
\end{equation}
\end{lemma}

\begin{proof}
Let $\gamma_1, \gamma_2\in[0,1]$. If $x'=\gamma_1x$, $y'=\gamma_2y$,
$x\in A$, and $y\in B$, then
$\left\|\frac{x+y}{2}\right\|\le\sqrt{1-\frac{\eps^2}{4}}\le1-\frac{\eps^2}{8}$.
Suppose that $\gamma_2=1$, then by triangle inequality
\begin{eqnarray*}
\left\|\frac{x'+y'}{2}\right\|=\left\|\frac{\gamma_1x+y}{2}\right\|
&\le&\gamma_1\left\|\frac{x+y}{2}\right\|+(1-\gamma_1)\left\|\frac{y}{2}\right\| \\
&\le&\gamma_1\left(1-\frac{\eps^2}{8}\right)+(1-\gamma_1)\frac{1}{2} \\
&\le&1-\frac{\eps^2}{8}.
\end{eqnarray*}
\end{proof}

Lemma~\ref{BMineq} and Eq.~(\ref{eq.euclid.dist}) directly imply the following measure concentration estimate for the sphere.

\begin{theorem} \label{measureconcen}
Let $A\subseteq\partial\cl{S}^d$ be a measurable set with $\pr[A]\ge\frac{1}{2}$. Suppose that $A_\eps$ denotes the $\eps$-neighborhood of $A$, i.e., the set of all $x\in\cl{S}^d$ whose Euclidean norm to $A$ is at most $\eps$. Then
\begin{equation} \label{eq.measureconcen}
1-\pr[A_\eps]\le2\exp\left(-\frac{\eps^2d}{4}\right).
\end{equation}
\end{theorem}

\begin{proof}
By Lemma~\ref{euclid.dist}, the set $\frac{A'+B'}{2}$ is contained in the sphere of radius $1-\frac{\eps^2}{8}$. By the Brun-Minkowski inequality; Lemma~\ref{BMineq}, and $\pr[A]\ge\frac{1}{2}$,
\begin{eqnarray}
\left(1-\frac{\eps^2}{8}\right)^d
&\ge&\mu\left(\frac{1}{2}(A'+B')\right)\ge\sqrt{\mu(A')\mu(B')} \\
&=&\sqrt{\pr[A]\pr[B]}\ge\sqrt{\frac{1}{2}\pr[B]}.
\end{eqnarray}
Thus, $\pr[B]\le2\left(1-\frac{\eps^2}{8}\right)^{2d}\le2\exp\left(-\frac{\eps^2d}{4}\right)$.
\end{proof}

We now record L\'{e}vy's theorem on the $d-1$-dimensional shell $\partial\cl{S}^{d}$ via concentration of Lipschitz functions. Let $f$ be a function $f:\partial\cl{S}^d\to\bb{R}$ defined on a sphere $\cl{S}^d$, and suppose that $f$ is 1-Lipschitz. A median $m_f$ of $f$ is any real number such that
\begin{equation} \label{eq.expect1}
\pr[f<m_f]\le\frac{1}{2}~~~~\mathrm{and}~~~~
\pr[f>m_f]\le\frac{1}{2}.
\end{equation}
The corresponding mean form follows from the standard comparison between medians and expectations for Lipschitz functions under a concentration bound.

\begin{theorem} \label{1lipschitz}
Let a function $f:\partial\cl{S}^d\to\bb{R}$ be 1-Lipschitz. Then, for all $\eps\in[0,1]$,
\begin{equation} \label{eq.1lipschitz}
\pr[|f-m_f|>\eps]\le C_1\exp\left(-c d\eps^2\right),
\end{equation}
where $C_1,c>0$ are universal constants.
\end{theorem}

\begin{proof}
Suppose that $A_-:=\{x\in\partial\cl{S}^d|f(x)\le m_f\}$ and $A_+:=\{x\in\partial\cl{S}^d|f(x)\ge m_f\}$. Then $\pr[A_-]\ge\frac{1}{2}$ and $\pr[A_+]\ge\frac{1}{2}$. Since $f$ is 1-Lipschitz, the complement of $(A_-)_\eps$ is contained in $\{x:f(x)>m_f+\eps\}$, and the complement of $(A_+)_\eps$ is contained in $\{x:f(x)<m_f-\eps\}$. By Theorem~\ref{measureconcen},
\begin{equation}
\pr[|f-m_f|>\eps]\le4\exp\left(-\frac{\eps^2d}{4}\right),
\end{equation}
which has the asserted form after renaming the universal constants.
\end{proof}

The L\'{e}vy's theorem also can be proved by classical isoperimetric inequality in Appendix I of Milman and Schechtman~\cite{MS}.

Let $(\cl{M},d_{\cl{M}},\mu)$ be a compact metric space $(\cl{M},d_{\cl{M}})$ with a Borel probability measure $\mu$.
Let $A$ be a set in the metric space $(\partial\cl{S}^d,d_{\partial\cl{S}^d})$ with a {\it geodesic} metric $d_{\cl{S}}(:=d_{\partial\cl{S}^d})$ on $\partial\cl{S}^d$ and $\delta\in[0,1]$. Define a subset $A_\delta:=\{x|d_{\cl{S}}(x,A)\le\delta\}$ from $A$. Then, for some fixed $x_0\in\partial\cl{S}^d$, $A_0=\{x|d_{\cl{S}}(x,x_0)\le r\}$ with $\mu(A_0)=a$. Then a classical isoperimetric inequality is following

\begin{theorem}
Let $A$ be a set in the metric space $(\partial\cl{S}^d,d_{\cl{S}},\mu)$ with a geodesic metric $d_{\cl{S}}$, and $\delta\in[0,1]$.
For each $a\in[0,1]$, there exists
\begin{equation}
\mu_\delta=\min\{\mu(A_\delta)|A\subseteq\partial\cl{S}^d, \mu(A)=a\},
\end{equation}
and it is obtained at $A_0$.
\end{theorem}

If $a=1/2$, then the following corollary is true.

\begin{cor} \label{isoperim}
Let $A$ be a set in the metric space $(\partial\cl{S}^d,d_{\cl{S}},\mu)$ with a geodesic metric $d_{\cl{S}}$, and $\delta\in[0,1]$.
Suppose that $\mu(A)\ge\frac{1}{2}$, then
\begin{equation}
\mu(A_\delta)\ge1-\sqrt{\frac{\pi}{8}}\exp\left(-\frac{\delta^2(d-2)}{2}\right).
\end{equation}
\end{cor}

\begin{proof}
Let $l_\theta$ be a set of all points on $\partial\cl{S}^d$ which are of distance $\theta$ from $A_{\frac{\pi}{2}}:=\{A|\mu(A)=\frac{1}{2}\}$ for some fixed $x_0\in\partial\cl{S}^d$. Suppose that $\partial\cl{S}^d$ has a radius 1, then the radius $\cos^{d-2}\theta$ of $l_\theta$ is proportional to $(d-2)$-dimensional volume of $l_\theta$. Since $A_\delta=\int_{-\pi/2}^\delta\cos^{d-2}d\theta/\int_{-\pi/2}^{\pi/2}\cos^{d-2}d\theta$ and hence define $\iota_{d-1}=\int_0^{\pi/2}\cos^{d-2}d\theta$. Let make use of the change of variables $\theta=\frac{1}{\sqrt{d-2}}\phi$, $\phi=\psi+\delta\sqrt{d-2}$ and the inequality $\cos\tau\le\exp(-\frac{\tau^2}{2})$. Then,
\begin{eqnarray*}
1-\mu(A_\delta)
&=&\frac{\int_{\delta\sqrt{d-2}}^{\frac{\pi\sqrt{d-2}}{2}}
\cos^{d-2}\frac{\phi}{\sqrt{d-2}}}
{2\sqrt{d-2}\iota_{d-2}}d\phi \\
&\le&\exp\left(-\frac{\delta^2(d-2)}{2}\right)
\frac{\int_0^{(\frac{\pi}{2}-\delta)\sqrt{d-2}}
\exp\left(-\frac{\psi^2}{2}\right)}{2\sqrt{d-2}\iota_{d-2}}d\psi \\
&\le&\sqrt{\frac{\pi}{2}}\frac{1}
{2\sqrt{d-2}\iota_{d-2}}\exp\left(-\frac{\delta^2(d-2)}{2}\right),
\end{eqnarray*}
where the last inequality is used Gaussian integral formula. Since $\iota_{t}=\frac{t-1}{t}\iota_{t-2}$ and hence $\sqrt{d-2}\iota_{d-2}\ge\sqrt{d-4}\iota_{d-4}\ge\min(\iota_1,\sqrt{2}\iota_2)=1$. Thus, we have
\begin{equation}
1-\mu(A_\delta)\le\sqrt{\frac{\pi}{8}}\exp\left(-\frac{\delta^2(d-2)}{2}\right).
\end{equation}
\end{proof}

Let $f$ be a continuous function $f:\partial\cl{S}^d\to\bb{R}$ defined on a hypersphere $\cl{S}^d$. Let $m_f$ be a median of $f$, so that
\begin{equation}
\mu\{x\in\partial\cl{S}^d |f\le m_f\}\ge\frac{1}{2}
~~\mathrm{and}~~
\mu\{x\in\partial\cl{S}^d |f\ge m_f\}\ge\frac{1}{2}.
\end{equation}

\begin{cor} \label{b.levy} (L\'{e}vy)
Let $f$ be a continuous function $f:\partial\cl{S}^d\to\bb{R}$ on a sphere $\cl{S}^d$. Suppose that $A=\{x\in\partial\cl{S}^d|f(x)\le m_f\}$, then
\begin{equation}
\mu(A_\delta)\ge1-\sqrt{\frac{\pi}{8}}\exp
\left(-\frac{\delta^2(d-2)}{2}\right).
\end{equation}
\end{cor}

\begin{widetext}
\begin{proof}
For the two median half-spaces $A_-=\{x:f(x)\le m_f\}$ and $A_+=\{x:f(x)\ge m_f\}$, Corollary~\ref{isoperim} gives the corresponding upper and lower tail estimates. Thus,
\begin{equation}
\mu((A_-)_\delta)\ge1-\sqrt{\frac{\pi}{8}}\exp\left(-\frac{\delta^2(d-2)}{2}\right)
~~~~\mathrm{and}~~~~
\mu((A_+)_\delta)\ge1-\sqrt{\frac{\pi}{8}}\exp\left(-\frac{\delta^2(d-2)}{2}\right).
\end{equation}
\end{proof}
\end{widetext}

\section{Proof of main results} \label{sec.main}
In this section, two kinds of proof of Theorem~\ref{antipodal} will be presented: The concentration of 1-Lipschitz function and the isoperimetric inequality.

Let $\cl{S}^d$ be a $d$-dimensional hypersphere and $\cl{E}^{d-1}\subset\cl{S}^d$ be a $d-1$-dimensional hyperdisk through the origin. (Note that $\partial\cl{E}^{d-1}:=E^{d-2}$.) Let $A$ be a nonempty compact subset of $\partial\cl{E}^{d-1}$. Suppose that $A':=\{\gamma' x|x\in A, \gamma'\in[0,1]\}$ in $\cl{E}^{d-1}$, then
\begin{equation}
\pr[A]=\mu(A'):=\frac{\vol(A')}{\vol(\cl{E}^{d-1})}.
\end{equation}
Let $B$ be a set of $\partial\cl{E}^{d-1}\backslash(A\cup A_\eps)$, where $A_\eps$ denotes the $\eps$-neighborhood of $A$ such that the set of all $x\in\partial\cl{E}^{d-1}$ whose Euclidean distance $\|\cdot\|$ to $A$ is at most $\eps$. Suppose that $\eps\in[0,1]$ and $\pr[A]\ge1/2$, then, for all $a\in A$ and $b\in B$, $\|a-b\|\ge\eps$. Let $A'$ and $B'$ are subsets in $\cl{E}^{d-1}$. Then, by Lemma~\ref{euclid.dist}, we also obtain $\left\|\frac{x'+y'}{2}\right\|\le1-\frac{\eps^2}{8}$ for any $x'\in A'$ and $y'\in B'$.

\begin{prop} \label{MC.equator}
Let $A\subseteq\partial\cl{E}^{d-1}$ be a measurable set with $\pr[A]\ge\frac{1}{2}$. Suppose that $A_\eps$ denotes the $\eps$-neighborhood of $A$, i.e., the set of all $x\in\cl{E}^{d-1}$ whose Euclidean distance to $A$ is at most $\eps$. Then
\begin{equation} \label{eq.MC.equator}
1-\pr[A_\eps]\le2\exp\left(-\frac{\eps^2{(d-1)}}{4}\right).
\end{equation}
\end{prop}

\begin{proof}
By above state (or Lemma~\ref{euclid.dist}), the set $\frac{A'+B'}{2}$ is contained in the hyperdisk of radius $1-\frac{\eps^2}{8}$. By the Brun-Minkowski inequality; Lemma~\ref{BMineq}, and $\pr[A]\ge\frac{1}{2}$, 
\begin{eqnarray}
\left(1-\frac{\eps^2}{8}\right)^{d-1}
&\ge&\mu\left(\frac{1}{2}(A'+B')\right)\ge\sqrt{\mu(A')\mu(B')} \nonumber\\
&=&\sqrt{\pr[A]\pr[B]}\ge\sqrt{\frac{1}{2}\pr[B]}.
\end{eqnarray}
Thus, $\pr[B]\le2\left(1-\frac{\eps^2}{8}\right)^{2(d-1)}\le2\exp\left(-\frac{\eps^{2}(d-1)}{4}\right)$.
\end{proof}

We now prove our main theorem on the $d-2$-dimensional hyperequator $\partial\cl{E}^{d-1}$ via concentration of Lipschitz functions. Let $f:\partial\cl{E}^{d-1}\to\bb{R}$ be a $\kappa$-Lipschitz function and let $m_f$ be a median of $f$. Then we have

\begin{theorem} \label{antipodal'}
Let $f:E^{d-2}\to \bb{R}$ be a $\kappa$-Lipschitz function on $E^{d-2}=\partial\cl{E}^{d-1}$. For every $\eps>0$, if a point $x\in E^{d-2}$ is chosen uniformly at random, then
\begin{equation} \label{eq.antipodal'}
\pr\left[|f(x)-m_f|\ge\eps\right]\le C_1\exp\left(\frac{-c(d-1)\eps^2}{\kappa^2}\right),
\end{equation}
where $m_f$ is a median of $f$ and $C_1,c>0$ are universal constants.
\end{theorem}

\begin{proof}
Suppose that $A_-=\{x\in\partial\cl{E}^{d-1}|f(x)\le m_f\}$ and $A_+=\{x\in\partial\cl{E}^{d-1}|f(x)\ge m_f\}$. Then $\pr[A_-]\ge\frac{1}{2}$ and $\pr[A_+]\ge\frac{1}{2}$. If $x\in (A_-)_{\eps/\kappa}$, then $f(x)\le m_f+\eps$; similarly, if $x\in (A_+)_{\eps/\kappa}$, then $f(x)\ge m_f-\eps$. By Proposition~\ref{MC.equator},
\begin{equation}
\pr[|f(x)-m_f|\ge\eps]\le4\exp\left(-\frac{(d-1)\eps^2}{4\kappa^2}\right),
\end{equation}
which has the asserted form after renaming the universal constants.
\end{proof}

Above theorem also can be proved by (classical) isoperimetric inequality as in Section~\ref{sec.levy}. Let $(\cl{M},d_{\cl{M}},\mu)$ be a compact metric space $(\cl{M},d_{\cl{M}})$ with a uniform probability measure $\mu$. Let $A$ be a set in the metric space
$(\partial\cl{E}^{d-1},d_{\partial\cl{E}^{d-1}})$ with a {\it geodesic} metric $d_{\cl{E}}(:=d_{\partial\cl{E}^{d-1}})$ on $\partial\cl{E}^{d-1}$, and $\delta\in[0,1]$. We define a subset $A_\delta:=\{x|d_{\cl{E}}(x,A)\le\delta\}$ from $A$. Then, for some fixed $x_0\in\partial\cl{E}^{d-1}$, $A_0=\{x|d_{\cl{E}}(x,x_0)\le r\}$ with $\mu(A_0)=a$. Then an isoperimetric inequality states that let $A$ be a set in the metric space $(\partial\cl{E}^{d-1},d_{\cl{E}},\mu)$ with a geodesic metric $d_{\cl{E}}$, and $\delta\in[0,1]$. For each $a\in[0,1]$, there exists $\mu_\delta=\min\{\mu(A_\delta)|A\subseteq\partial\cl{E}^{d-1}, \mu(A)=a\}$, and it is obtained at $A_0$.

Suppose that $a=1/2$, then, by Corollary~\ref{isoperim}, we have following statements: Let $A$ be a set in the metric space $(E^{d-2},d_{E},\mu)$ with a geodesic metric $d_{E}$, and $\delta\in[0,1]$. Suppose that $\mu(A)\ge\frac{1}{2}$, then
\begin{equation} \label{eq.isoperim'}
\mu(A_\delta)\ge1-\sqrt{\frac{\pi}{8}}\exp\left(-\frac{\delta^2(d-1)}{2}\right).
\end{equation}

Let $f$ be a continuous function $f:E^{d-2}\to\bb{R}$ defined on a hyperdisk $\cl{E}^{d-1}$. Let $m_f$ be a median of $f$, so that
\begin{eqnarray}
&\mu\{x\in E^{d-2} |f\le m_f\}\ge\frac{1}{2} \nonumber\\
&~~~~\mathrm{and}~~~~
\mu\{x\in E^{d-2} |f\ge m_f\}\ge\frac{1}{2}.
\end{eqnarray}

\begin{prop} \label{c.levy} (L\'{e}vy)
Let $f$ be a continuous function $f:E^{d-2}\to\bb{R}$ on a fixed hyperdisk $\cl{E}^{d-1}$. Suppose that $A=\{x\in E^{d-2}|f(x)\le m_f\}$, then
\begin{equation}
\mu(A_\delta)\ge1-\sqrt{\frac{\pi}{8}}\exp\left(-\frac{\delta^2(d-1)}{2}\right).
\end{equation}
\end{prop}

\begin{widetext}
\begin{proof}
For the two median half-spaces $A_-=\{x:f(x)\le m_f\}$ and $A_+=\{x:f(x)\ge m_f\}$, Eq.~(\ref{eq.isoperim'}) gives
\begin{equation}
\mu((A_-)_\delta)\ge1-\sqrt{\frac{\pi}{8}}\exp\left(-\frac{\delta^2(d-1)}{2}\right)
~~~~\mathrm{and}~~~~
\mu((A_+)_\delta)\ge1-\sqrt{\frac{\pi}{8}}\exp\left(-\frac{\delta^2(d-1)}{2}\right).
\end{equation}
This completes the proof.
\end{proof}
\end{widetext}

Note that Proposition~\ref{c.levy} gives the same L\'{e}vy-type form as Theorem~\ref{antipodal'}. Finally, we record the precise sense in which the antipodal great subsphere is the last set in the coordinate localization scheme considered here.

\begin{remark} \label{finalCM}
The statement is not a nonexistence theorem for subsets of a sphere. Rather, once $\cl{A}^{d-3}=E^{d-2}\cap H_{d-2}^0$ has been fixed, any further concentration statement must specify an additional metric-measure space and a new concentration function. Thus the present argument identifies a terminal object only for the chosen sequence of coordinate neighborhoods.
\end{remark}

\section{Concentration of measure for hypersphere} \label{sec.CM}
This section summarizes concentration estimates for the unit ball and its boundary sphere, following the standard presentation of high-dimensional geometry in Hopcroft and Kannan~\cite{HK}. Let $\cl{S}^d$ be the closed $d$-dimensional unit ball in $\bb{R}^d$ and let $\partial\cl{S}^d$ be its boundary sphere.

\begin{lemma} \label{areavol}
Let $\cl{S}^d$ be the $d$-dimensional unit ball. The boundary area and the volume of the ball are given by
\begin{equation}
\area(\partial\cl{S}^d)=\frac{2\pi^{\frac{d}{2}}}{\Gamma\left(\frac{d}{2}\right)}
~~~~\tn{and}~~~~
\vol(\cl{S}^d)=\frac{\pi^{\frac{d}{2}}}{\frac{d}{2}\Gamma\left(\frac{d}{2}\right)}.
\end{equation}
\end{lemma}

\begin{proof}
The volume of the unit ball $\cl{S}^d$ is given by
\begin{eqnarray}
\vol(\cl{S}^d)
&=&\int_{\partial\cl{S}^d}\int_{r=0}^1r^{d-1}d\Omega dr \nonumber\\
&=&\frac{1}{d}\int_{\partial\cl{S}^d}d\Omega=\frac{\area(\partial\cl{S}^d)}{d}.
\end{eqnarray}
Consider the following Gaussian integral in Cartesian coordinates
\begin{equation}
\iota(\cl{S}^d):=\left[\int_{-\infty}^\infty\exp(-x^2)dx\right]^d=\pi^{\frac{d}{2}}.
\end{equation}
Then, in polar coordinate,
\begin{eqnarray}
\iota(\cl{S}^d)&=&\int_{\partial\cl{S}^d}d\Omega\int_0^\infty\exp(-r^2)r^{d-1}dr \nonumber\\
&=&\frac{\area(\partial\cl{S}^d)}{2}\int_0^\infty\exp(-t)t^{\frac{d}{2}-1}dt \nonumber\\
&=&\frac{\area(\partial\cl{S}^d)}{2}\Gamma\left(\frac{d}{2}\right).
\end{eqnarray}
This completes the proof.
\end{proof}

\begin{cor} \label{VOLarea}
Let $\cl{S}^d$ be a $d$-dimensional hypersphere with unit radius. Suppose that the dimension $d$ is sufficiently large, then
\begin{equation}
\lim_{d\to\infty}\vol(\cl{S}^d)=0.
\end{equation}
\end{cor}

\begin{proof}
By Stirling's formula,
\begin{equation}
\Gamma\left(\frac{d}{2}+1\right)\sim \sqrt{\pi d}\left(\frac{d}{2e}\right)^{d/2}.
\end{equation}
Since $\vol(\cl{S}^d)=\pi^{d/2}/\Gamma(d/2+1)$, the denominator eventually dominates the numerator and $\vol(\cl{S}^d)\to0$.
\end{proof}

The preceding computation is a statement about total Euclidean volume, not by itself a statement that a probability measure equals the boundary. The probabilistic form is the usual thin-shell statement:

\begin{cor} \label{cor.mu1}
Let $\mu_1$ be the uniform probability measure on the unit ball $\cl{S}^d$. For every fixed $\eta>0$,
\begin{equation}
\lim_{d\to\infty}\mu_1\left(\left\{x\in\cl{S}^d:\operatorname{dist}(x,\partial\cl{S}^d)\le\eta\right\}\right)=1.
\end{equation}
\end{cor}

The following lemma describes that the volume is {\it almost all} concentrated at a $d-1$- dimensional hyperdisk of the hypersphere $\cl{S}^d$.

\begin{lemma} \label{voleq}
Let $\cl{S}^d$ be a $d$-dimensional hypersphere with unit radius. Suppose that $C_1>0$. Then the fraction of the volume of the hemisphere above the hyperplane $x_1=\frac{C_1}{\sqrt{d-1}}$ is less than $\frac{2}{C_1}\exp\left(-\frac{C_1^2}{2}\right)$.
\end{lemma}

\begin{proof}
Fix the north pole on $x_1$ axis at $x_1=1$, and divide the hypersphere in half by a hyperplane $H_{d-1}^0$ i.e. a hyperdisk.
The disk with dimension $d-1$ is
\begin{equation}
D_0:=\cl{S}^d\cap H_{d-1}^0=\{x||x|\le1, x_1=0\}.
\end{equation}

Consider a $\eps$ neighborhood disk $D_\eps$ at $x_1=\eps$ parallel to $x_1=0$, and the portion of the hypersphere above
$D_\eps$ as $D_\eps^\uparrow:=\{x||x|\le1, x_1\ge\eps\}$. Thus the incremental volume is a disk of width $dx_1$ whose face is a hypersphere of dimension $d-1$ with radius $\sqrt{1-x_1^2}$. So the surface area of the disk is
\begin{equation}
\area(D_{x_1\ge\eps})=(1-x_1^2)^{\frac{d-1}{2}}\vol(\cl{S}^{d-1}).
\end{equation}
By using $1+a\le\exp(a)$, $a\in\bb{R}$, and $\frac{x_1}{\eps}\ge1$,
\begin{eqnarray}\label{eq.diskupper}
\vol(D_\eps^\uparrow)
&=&\vol(\cl{S}^{d-1})\int_\eps^1(1-x_1^2)^{\frac{d-1}{2}}dx_1 \nonumber\\
&\le&\vol(\cl{S}^{d-1})\int_{\eps}^\infty\exp\left(-\frac{d-1}{2}x_1^2\right)dx_1 \nonumber\\
&\le&\vol(\cl{S}^{d-1})\int_{\eps}^\infty\frac{x_1}{\eps}\exp\left(-\frac{d-1}{2}x_1^2\right)dx_1 \nonumber\\
&=&\frac{1}{\eps(d-1)}\exp\left(-\frac{d-1}{2}\eps^2\right)
\vol(\cl{S}^{d-1}).
\end{eqnarray}

Now define a upper cylinder as $C_{x_1}^\uparrow:=\{x||x|\le1, r=\sqrt{1-x_1^2}, h=x_1\}$. Then the lower bound of the volume of the entire upper hemisphere is given by
\begin{eqnarray} \label{eq.disklower}
\vol(D_0^\uparrow)&\ge&\vol(C_{x_1=\frac{1}{\sqrt{d-1}}}) \nonumber\\
&=&\vol(D_{x_1=\frac{1}{\sqrt{d-1}}})\frac{1}{\sqrt{d-1}} \nonumber\\
&=&\frac{1}{\sqrt{d-1}}\vol(\cl{S}^{d-1})\left(1-\frac{1}{d-1}\right)^{\frac{d-1}{2}} \nonumber\\
&\ge&\frac{1}{\sqrt{d-1}}\vol(\cl{S}^{d-1})\left(1-\frac{1}{d-1}\frac{d-2}{2}\right) \nonumber\\
&=&\frac{1}{2\sqrt{d-1}}\vol(\cl{S}^{d-1}),
\end{eqnarray}
where the last inequality follows from $(1-x)^a\ge1-ax$. From the Eqs.~(\ref{eq.diskupper}) and (\ref{eq.disklower}), the fraction of the volume is
\begin{equation}
f_V:=\frac{\vol(D_\eps^\uparrow)}{\vol(D_0^\uparrow)}
\le\frac{2}{\eps\sqrt{d-1}}\exp\left(-\frac{d-1}{2}\eps^2\right).
\end{equation}
Substitute $\eps$ to $\frac{C_1}{\sqrt{d-1}}$, then $f_V\le\frac{2}{C_1}\exp\left(-\frac{C_1^2}{2}\right)$.
\end{proof}

Above Lemma~\ref{voleq} states that most of the volume of a unit ball lies in a slab of width $O(1/\sqrt{d})$ around a fixed central hyperdisk. In neighborhood language:

\begin{cor} \label{cor.mu2}
Let $\mu_2$ be the uniform probability measure on $\cl{S}^d$ and let $\cl{E}^{d-1}=\cl{S}^d\cap H_{d-1}^0$ be a fixed central hyperdisk. For every fixed $\eta>0$,
\begin{equation}
\lim_{d\to\infty}\mu_2\left(\left\{x\in\cl{S}^d:\operatorname{dist}(x,\cl{E}^{d-1})\le\eta\right\}\right)=1.
\end{equation}
\end{cor}

The following lemma considers a volume of the hypersphere near a narrow hyperannulus.

\begin{lemma} \label{volannu}
Let $\cl{S}^d$ be a $d$-dimensional hypersphere with unit radius. Then, for sufficiently large $d$, the ratio $R$ of the volume of a hypersphere of radius $1-\eta$ to the volume of a unit hypersphere goes to 0:
\begin{equation}
\lim_{d\to\infty}\frac{\vol(\cl{S}_{r=1-\eta}^d)}{\vol(\cl{S}^d)}\to0.
\end{equation}
\end{lemma}

\begin{proof}
\begin{equation}
R:=\frac{(1-\eta)^d\vol(\cl{S}^d)}{\vol(\cl{S}^d)}=(1-\eta)^d\le\exp(-\eta d).
\end{equation}
\end{proof}

Let $C$ be a sufficiently large constant. Suppose that $\eta=C/d$. Then all but approximately $\exp(-C)$ of the volume of the unit ball is contained in a thin annulus of width $C/d$ near the boundary. The following lemma is concerned with surface area near the equator, namely the geometric form of L\'{e}vy's theorem.

\begin{lemma} \label{surfeq}
Let $\cl{S}^d$ be the unit ball and let $\partial\cl{S}^d$ be its boundary sphere. For any $C_2>0$, the fraction of the surface area above the plane $x_1=\frac{C_2}{\sqrt{d-2}}$ is bounded by $C\exp(-cC_2^2)$ for universal constants $C,c>0$.
\end{lemma}

\begin{proof}
Fix also the north pole on $x_1$ axis at $x_1=1$. Define
\begin{equation}
\partial D_0:=\partial\cl{S}^d\cap H_{d-1}^0=\{x||x|=1, x_1=0\}=E^{d-2}.
\end{equation}
For a point chosen uniformly from $\partial\cl{S}^d$, the first coordinate has density proportional to $(1-t^2)^{(d-3)/2}$ on $[-1,1]$. Hence, for $0<\eps<1$,
\begin{equation} \label{eq.surfup}
\mu\{x\in\partial\cl{S}^d:x_1\ge\eps\}
\le C\exp\left(-c(d-2)\eps^2\right)
\end{equation}
with universal constants $C,c>0$. Substituting $\eps=C_2/\sqrt{d-2}$ gives the claim.
\end{proof}

Above Lemma~\ref{surfeq} states that most of the area of the boundary $\partial\cl{S}^d$ of radius 1 lies in a small neighborhood of an arbitrary hyperequator $E^{d-2}$. So, we can assert

\begin{cor} \label{cor.mu3}
Let $\mu_3$ be the uniform probability measure on $\partial\cl{S}^d$, and let $E^{d-2}:=\partial\cl{E}^{d-1}$ be a fixed hyperequator. For every fixed $\eta>0$,
\begin{equation}
\lim_{d\to\infty}\mu_3\left(\left\{x\in\partial\cl{S}^d:\operatorname{dist}(x,E^{d-2})\le\eta\right\}\right)=1.
\end{equation}
\end{cor}

\begin{prop} \label{prop.arcanti}
Let $E^{d-2}$ be a fixed hyperequator with uniform surface measure. Then, for any $C_3>0$ and for some fixed coordinate $x_2$ on $E^{d-2}$, the fraction of the measure above the plane $x_2=\frac{C_3}{\sqrt{d-3}}$ is bounded by $C\exp(-cC_3^2)$ for universal constants $C,c>0$.
\end{prop}

\begin{proof}
Fix the north pole on $x_1$ axis at $x_1=1$, and choose another north pole $x_2=1$ on the hyperequator $E^{d-2}$
given by Lemma~\ref{surfeq}. Define
\begin{equation*}
\partial L_0:=E^{d-2}\cap H_{d-2}^0
=\{x||x|=1, x_1=x_2=0\}=\cl{A}^{d-3}.
\end{equation*}

For a point chosen uniformly from $E^{d-2}$, the coordinate $x_2$ has the same one-coordinate concentration form as in Lemma~\ref{surfeq}, with dimension parameter $d-3$. Therefore,
\begin{equation} \label{eq.arcup}
\mu\{x\in E^{d-2}:x_2\ge\eps\}\le C\exp\left(-c(d-3)\eps^2\right).
\end{equation}
Substituting $\eps=C_3/\sqrt{d-3}$ gives the asserted bound.
\end{proof}

Above Proposition~\ref{prop.arcanti} gives the following neighborhood formulation:

\begin{cor} \label{cor.mu4}
Let $E^{d-2}$ be a $d-2$-dimensional hyperequator and let $\mu_4$ be the uniform probability measure on $E^{d-2}$. For every fixed $\eta>0$,
\begin{equation}
\lim_{d\to\infty}\mu_4\left(\left\{x\in E^{d-2}:\operatorname{dist}(x,\cl{A}^{d-3})\le\eta\right\}\right)=1,
\end{equation}
where $\cl{A}^{d-3}=E^{d-2}\cap H_{d-2}^0$ is a codimension-one antipodal great subsphere of $E^{d-2}$.
\end{cor}

Let $\cl{S}^d$ be a $d$-dimensional hypersphere with unit radius. By Lemmas~\ref{areavol}, \ref{voleq} and \ref{surfeq}, and
Proposition~\ref{prop.arcanti}, Eq.~(\ref{eq.conj}) is supported in the neighborhood sense described above. This implies that the chosen localization scheme leads naturally to an antipodal great subsphere. By the Borsuk--Ulam theorem~\cite{LS}, the following corollary is satisfied.

\begin{cor}
Let $\cl{A}^{d-3}$ be an antipodal great subsphere, which is homeomorphic to $S^{d-3}$. Suppose that $g:\cl{A}^{d-3}\to\bb{R}^{d-3}$ is continuous. Then there exists a point $x\in\cl{A}^{d-3}$ such that $g(x)=g(-x)$.
\end{cor}

\section{Conclusion} \label{sec.conclusion}
We have revisited L\'{e}vy's lemma from the viewpoint of successive geometric localization on high-dimensional spheres.  The usual concentration statement says that Lipschitz functions on the sphere are almost constant with overwhelming probability, and it is commonly interpreted as concentration of surface measure near a median hyperequator.  In this work we separated the geometric steps behind this picture in neighborhood language: the volume of a high-dimensional ball is localized near its boundary, the boundary measure is localized near an arbitrary hyperequator, and the induced measure on the hyperequator is localized near a codimension-two antipodal great subsphere.

The main technical contribution is the hyper-equatorial concentration estimate in Theorem \ref{antipodal'}, which gives a L\'{e}vy-type bound for Lipschitz functions on $E^{d-2}=\partial\cl{E}^{d-1}$ with an exponential dependence on $d-1$.  The proof follows the classical architecture of L\'{e}vy's lemma: one derivation uses the Brunn--Minkowski inequality and Euclidean neighborhoods, while the other uses the corresponding isoperimetric formulation.  This parallel proof structure shows that the estimate is a dimension-shifted instance of the same high-dimensional spherical geometry that underlies the standard lemma.

The geometric estimates in Section \ref{sec.CM} support the conjectural localization scheme as
\begin{equation*}
\cl{S}^{d}\longrightarrow \partial\cl{S}^{d}\longrightarrow E^{d-2}\longrightarrow \cl{A}^{d-3}.
\end{equation*}
In this sense, the antipodal great subsphere appears as a natural terminal object for the coordinate localization procedure considered here.  This is also consistent with the topological role of antipodal pairs expressed by the Borsuk--Ulam theorem.  Since L\'{e}vy-type concentration estimates are widely used in quantum information theory, especially in the study of random pure states, random unitaries, typical entanglement, unitary designs, and quantum learning, a more precise understanding of the constant-level and geometric structure of these inequalities may be useful beyond the purely asymptotic regime.

Several questions remain open.  First, the constants in the exponent should be optimized and compared carefully with the sharp spherical isoperimetric constants.  Second, the limiting statement in Conjecture \ref{conj} should be strengthened in a fully measure-theoretic language, for instance in terms of concentration functions or metric-measure convergence.  Third, it would be useful to identify whether the antipodal localization described here gives quantitative improvements in concrete quantum-information applications where L\'{e}vy's lemma is currently used only up to universal constants. These questions indicate that the antipodal viewpoint may provide a productive refinement of the standard concentration-of-measure toolbox.

\section*{Acknowledgements}
This work was supported by the National Research Foundation of Korea (NRF) through a grant funded by the Ministry of Science and ICT (Grant No. RS-2025-00515537), the Institute for Information \& Communications Technology Promotion (IITP) grant funded by the Korean government (MSIP) (Grants Nos. RS-2025-02304540 and RS-2019-II190003), and the Korea Institute of Science and Technology Information (Grant No.P26028). 

\section*{Author Contributions}
K.J. conceived, suggested, wrote, and also reviewed the manuscript.

\section*{Data Availability}
No datasets were generated or analysed during the current study.

\section*{Conflict of Interest}
The author declares no conflict of interest.


%

\end{document}